\def\year{2017}\relax
\documentclass[letterpaper]{article}
\usepackage{aaai17}
\usepackage{times}
\usepackage{helvet}
\usepackage{courier}
\usepackage{algorithm}
\usepackage{algorithmic}
\usepackage{enumerate,setspace}
\usepackage{amsmath,amssymb,amsfonts,ascmac}
\usepackage{amsthm}
\usepackage{booktabs,microtype}
\usepackage{graphicx}
\newcommand{\sfGASP}{\scalebox{0.9}{\sf GASP}}
\newcommand{\gGASP}{\scalebox{0.9}{\sf gGASP}}

\newcommand{\calC}{\mathcal{C}}

\theoremstyle{plain}
	  \newtheorem{theorem}{Theorem}

\theoremstyle{definition}
	  \newtheorem{define}[theorem]{Definition}
	  
	  \newtheorem{example}[theorem]{Example}
	  
\theoremstyle{remark}
	  
\newcommand*{\citet}[1]{\citeauthor{#1} \shortcite{#1}}
\newcommand*{\citeNP}[1]{\citeauthor{#1} \citeyear{#1}}

\usepackage{aaai_hyperref}
 
\begin{document}
%
\title{Group Activity Selection on Social Networks}
\author{Ayumi Igarashi, Dominik Peters, Edith Elkind \\
Department of Computer Science \\
University of Oxford, UK \\
$\{$ayumi.igarashi, dominik.peters, edith.elkind$\}$@cs.ox.ac.uk}
\copyrightyear{2017}
\pagestyle{plain}
\nocopyright
\maketitle
\global\csname @topnum\endcsname 0
\global\csname @botnum\endcsname 0
\begin{abstract}
We propose a new variant of the group activity selection problem (\sfGASP), where the agents are placed on a social 
network and activities can only be assigned to connected subgroups. We show that if multiple groups
can simultaneously engage in the same activity, finding a stable outcome is easy as long as the network
is acyclic. In contrast, if each activity can be assigned to a single group only, finding stable outcomes 
becomes 
intractable, even if the underlying network is very simple:
the problem of determining whether 
a given instance of a \sfGASP\ admits a Nash stable outcome turns out to be NP-hard 
when the social network is a path, a star, or if the size of each connected component is bounded by a constant.
On the other hand, we obtain fixed-parameter tractability results for this problem with respect
to the number of activities.
\end{abstract}

\section{Introduction}
\noindent
Companies assign their employees to different departments, large decision-making bodies split their members into expert 
committees, and university faculty form research groups: division of labor, and thus group formation, is 
everywhere. For a given assignment of agents to activities (such as management, product development, or marketing) to be 
successful, two considerations are particularly important: the agents need to be capable to work on their activity, and 
they should be willing to cooperate with other members of their group.

Many relevant aspects of this setting are captured by the {\em group activity selection problem} (\sfGASP), 
introduced by Darmann et al.~\shortcite{Darmann2012}. In \sfGASP\ players have preferences over pairs of the form 
(activity, group size).
The intuition behind this formulation is that certain tasks are best performed in small or large groups, 
and agents may differ in their preferences over group sizes; however, they are indifferent about
other group members' identities. In the analysis of \sfGASP, desirable outcomes 
correspond to \emph{stable} and/or \emph{optimal} assignments of players to activities, i.e., assignments 
that are resistant to player deviations and/or
maximize the total welfare. In the work of Darmann et al.~\shortcite{Darmann2012}, players are assumed 
to have approval 
preferences, and a particular focus is placed on individually rational assignments 
with the maximum number of participants; subsequently, \citeauthor{Darmann2015}~\shortcite{Darmann2015}
investigated a model where players submit ranked ballots.

However, the basic model of \sfGASP\ ignores the relationships among the agents: 
Do they know each other? Are their working styles and personalities compatible?
Typically, we cannot afford to ask each agent about her preferences over all pairs of the form 
(coalition, activity), as the number of possible coalitions grows quickly with the number of agents.
A more practical alternative is to adopt the ideas of \citet{Myerson1977} and assume
that the relationships among the agents are encoded by a \emph{social network},
i.e., an undirected graph where nodes correspond to players and edges represent 
communication links between them; one can then require that each group is connected
with respect to this graph.

\begin{table*}[t]
	\centering
	\begin{tabular}{llccc}
		\toprule
		&& Complexity (general case) & few activities (FPT wrt $p$) & copyable activities \\
		\midrule
		Nash stability & trees & NP-c. (Thm~\ref{thm:hardness:path:NS}) && poly time (Thm~\ref{thm:NS:copyable}) \\
		& paths & NP-c. (Thm~\ref{thm:hardness:path:NS}) & $O^*(4^pn^{2})$ (Thm~\ref{thm:FPT:path:NS}) & poly time (Thm~\ref{thm:NS:copyable})\\
		& stars & NP-c. (Thm~\ref{thm:hardness:star:NS}) & $O^*(2^{p}p^{p+1}n\log n)$ (Thm~\ref{thm:FPT:star:NS}) & poly time (Thm~\ref{thm:NS:copyable})\\
		& small components & NP-c. (Thm~\ref{thm:hardness:smallcomponent:NS}) & $O^*(p^{c}8^{p}kn^2)$ (Thm~\ref{thm:FPT:smallcomponents:NS}) \\
		\midrule
		core stability & trees & NP-c. (Thm~\ref{thm:hardness:path:core}) & & poly time (Thm~\ref{thm:core:copyable}) \\
		& paths & NP-c. (Thm~\ref{thm:hardness:path:core}) & & poly time (Thm~\ref{thm:core:copyable}) \\
		& stars & NP-c. (Thm~\ref{thm:hardness:path:core}) & & poly time (Thm~\ref{thm:core:copyable}) \\
		& small components & NP-c. (Thm~\ref{thm:hardness:path:core}) &  $O^*(p^{c+1}8^{p}kn^2)$ (Thm~\ref{thm:FPT:smallcomponents:core}) \\
		\bottomrule
	\end{tabular}
	\caption{Overview of our complexity results. Here, $n$ is the number of players, $p$ is the number of activities, and $c$ is a bound on the size of the connected components. The NP-completeness results for small components hold even for $c = 4$ for Nash stability and for $c = 3$ for core stability.}
	\vspace{-5pt}
	\label{table}
\end{table*}

In this paper we extend the basic model of \sfGASP\ to take into account the agents' social network. 
We formulate several notions of stability for this setting, including Nash stability and core stability,
and study the complexity of computing stable outcomes in our model. 
These notions of stability are inspired by the hedonic 
games literature \cite{Aziz2016} and were applied in the 
\sfGASP\ setting by \citet{Darmann2012} and \citet{Darmann2015}.
 
Now, hedonic games on social networks were recently considered by \citet{Igarashi2016},
who showed that if the underlying network is acyclic, stable outcomes are guaranteed to exist
and some of the problems known to be computationally hard for the unrestricted setting
become polynomial-time solvable. We obtain a similar result for \sfGASP, but only
if several groups of agents can simultaneously engage in the same activity, i.e., 
if the activities are {\em copyable}. In contrast, we show that if each activity 
can be assigned to at most one coalition, finding a stable outcome is hard
even if the underlying network is very simple.
Specifically, checking the 
existence of Nash stable or core stable outcomes turns out to be NP-hard 
even for very restricted classes of graphs, including paths, stars, 
and graphs with constant-size connected components. We believe that this result is 
remarkable since, in the context of cooperative games, such restricted networks 
usually enable one to design efficient algorithms for computing stable solutions (see, e.g., 
\citeNP{Chalkiadakis2016}; \citeNP{Elkind2014}; 
\citeNP{Igarashi2016}).

Given these hardness results, we switch to the fixed parameter tractability paradigm.
In the context of \sfGASP, a particularly relevant parameter is the number of activities:
generally speaking, we expect the number of players to be considerably larger than the number
of available activities. We show that for the restricted classes of networks 
used in our hardness proofs (i.e., paths, stars, and graphs with small connected components)
finding a Nash stable outcome is in FPT with respect to the number of activities;
some of our results extend to the core stable outcomes and to somewhat more general networks
(though not to arbitrary networks).
Our results are summarized in Table~\ref{table}.

\section{Preliminaries}
For $s\in{\mathbb N}$, let $[s]=\{1,2,\ldots,s\}$.
An instance of the {\em Group Activity Selection Problem} (\sfGASP) is given by a finite set of {\em players} 
$N=[n]$, a finite set of {\em activities} $A=A^{*}\cup\{a_{\emptyset}\}$ where 
$A^{*}=\{a_{1},a_{2},\ldots,a_{p}\}$ and $a_{\emptyset}$ is the {\em void activity}, and a {\em profile} 
$(\succeq_{i})_{i \in N}$ of complete and transitive preference relations over the set of {\em alternatives} 
$X=A^{*} \times [n]\cup \{(a_{\emptyset},1)\}$. Intuitively, $a_\emptyset$ corresponds to staying 
alone and doing nothing; multiple agents can make that choice independently from each other.

We refer to subsets $S\subseteq N$ of players as {\em coalitions}. 
We say that two non-void activities $a$ and $b$ are {\em equivalent} if for every player $i \in N$ and every $\ell \in [n]$ 
it holds that $(a,\ell) \sim_{i} (b,\ell)$. A non-void activity $a \in A^{*}$ is called {\em copyable} if $A^{*}$ contains 
at least $n$ activities that are equivalent to $a$ (including $a$ itself). 
We say that player $i \in N$ {\em approves} an alternative $(a,k)$ if $(a,k) \succ_{i} (a_{\emptyset},1)$. 

An outcome of a \sfGASP\ is an {\em assignment} of activities $A$ to players $N$, i.e., a mapping $\pi:N \rightarrow A$.
Given an assignment $\pi:N \rightarrow A$ and a non-void activity $a \in A^{*}$,
we denote by $\pi^{a}=\{\, i \in N \mid \pi(i)=a \,\}$ the set of players 
assigned to $a$. Also, if $\pi(i)\neq a_\emptyset$,
we denote  by $\pi_{i}=\{i\}\cup\{\, j \in N \mid \pi(j)=\pi(i)\}$ 
the set of players assigned to the same activity as player $i \in N$;
we set $\pi_i=\{i\}$ if $\pi(i)= a_\emptyset$.
An assignment $\pi:N \rightarrow A$ of a \sfGASP\ is {\em individually rational} (IR) if 
for every player $i \in N$ with $\pi(i)\neq a_\emptyset$
we have $(\pi(i),|\pi_{i}|) \succeq_{i} (a_{\emptyset},1)$.
A coalition $S \subseteq N$ and an activity $a \in A^{*}$ {\it strongly block} an assignment $\pi:N \rightarrow A$ if 
$\pi^a\subseteq S$ and $(a,|S|) \succ_{i} (\pi(i),|\pi_{i}|)$ for all $i \in S$.
An assignment $\pi:N \rightarrow A$ of a \sfGASP\ is called {\em core stable} (CR) if it is individually rational, and 
there is no coalition $S \subseteq N$ and activity $a \in A^{*}$ 
such that $S$ and $a$  strongly block $\pi$.
Given an assignment $\pi:N \rightarrow A$ of a \sfGASP, a player $i \in N$ is said to have 
an {\em NS-deviation} 
to activity $a \in A^{*}$
if $(a,|\pi^{a}|+1) \succ_{i} (\pi(i),|\pi_{i}|)$, that is, if $i$ would prefer to join the group $\pi^{a}$.
An assignment $\pi:N \rightarrow A$ of a \sfGASP\ is called {\em Nash stable} (NS) 
if it is individually rational and no player $i\in N$ has an NS-deviation 
to some $a \in A^{*}$.

\section{Our Model}
We now define a group activity selection problem where communication structure among players is restricted by an undirected graph.
\begin{define}
An instance of the {\em Group Activity Selection Problem with graph structure} (\gGASP) is given by an instance 
$(N,(\succeq_{i})_{i \in N},A)$ of a \sfGASP\ and a set of communication links between players $L \subseteq \{\, \{i,j\} 
\mid i,j\in N \land i\neq j \,\}$.
\end{define}

A coalition $S \subseteq N$ is said to be {\em feasible} if $S$ is connected in the graph $(N,L)$. An outcome of a 
\gGASP\ is a {\em feasible assignment} $\pi:N \rightarrow A$ such that $\pi_{i}$ is a feasible 
coalition for every $i \in N$. We adapt the definitions of stability concepts to our setting as follows. We say that a 
deviation by a group of players is feasible if the deviating coalition itself is feasible; 
a deviation by an individual player where player $i$ 
joins activity $a$ is feasible if $\pi^{a}\cup \{i\}$ is feasible. We modify the definitions in the previous section by 
only requiring stability against feasible deviations.
Note that an ordinary \sfGASP\ (without graph structure) is equivalent to a \gGASP\ 
where the underlying graph $(N,L)$ is complete.

In this paper, we will be especially interested in \gGASP s where $(N,L)$ is \emph{acyclic}. This restriction 
guarantees the existence of stable outcomes in many other cooperative game settings. However, this is not the case 
for \gGASP s: here, both core and Nash stable outcomes may fail to exist, even if $(N,L)$ is a path or a star.

\begin{example}\label{ex:core:empty}
Consider a \gGASP\ with $N=\{1,2,3\}$, $A^{*}=\{a,b\}$, $L=\{\{1,2\},\{2,3\}\}$, where preferences $(\succeq_{i})_{i \in N}$ are given as follows:
\begin{align*}
1:&~ (b,2) \succ_{1} (a,3) \succ_{1} (a_{\emptyset},1)\\
2:&~ (a,2) \succ_{2} (b,2) \succ_{2} (a,3) \succ_{2} (a_{\emptyset},1)\\
3:&~ (a,3) \succ_{3} (b,1) \succ_{3} (a,2) \succ_{3} (a_{\emptyset},1)
\end{align*}
There are only four individually rational feasible assignments; in each case, $\pi$ admits a strongly blocking feasible 
coalition and activity. First, when $\pi(1)=b$, $\pi(2)=b$, $\pi(3)=a_{\emptyset}$, the coalition $\{2,3\}$ together with 
activity $a$ strongly blocks~$\pi$. Second, when $\pi(1)=a_{\emptyset}$, $\pi(2)=a$, $\pi(3)=a$, the coalition $\{3\}$ 
together with activity $b$ strongly blocks $\pi$. Third, when $\pi(1)=a_{\emptyset}$, $\pi(2)=a_{\emptyset}$, $\pi(3)=b$, 
the coalition $\{1,2,3\}$ together with activity $a$ strongly blocks $\pi$. Finally, when $\pi(1)=a$, $\pi(2)=a$, and $\pi(3)=a$, the coalition $\{1,2\}$ together with activity $b$ strongly blocks $\pi$. \qed
\end{example}

Similarly, a Nash stable outcome is not guaranteed to exist even for \gGASP s on paths.

\begin{example}[Stalker game]\label{ex:NS:empty}
Consider a two-player \gGASP\ 
where player 1 is happy to participate in any activity as long as she is alone,
and player 2 always wants to participate in an activity with player 1.
This instance admits no Nash stable outcomes: if player 1 engages in an activity,
then player 2 wants to join her coalition, causing player 1 to deviate
to another (possibly void) activity. \qed
\end{example}

However, if all activities are copyable, we can effectively 
treat \gGASP\ as a hedonic game on a graph. In particular, 
we can invoke a famous result of \citet{Demange2004} concerning the stability of non-transferable 
utility games on trees. Thus, requiring all activities to be copyable 
allows us to circumvent the non-existence result for the core 
(Example~\ref{ex:core:empty}). The argument is constructive.

\begin{theorem}[implicit in the work of~\citeNP{Demange2004}]\label{thm:core:copyable}
For every \gGASP\ where each activity $a \in A^{*}$ is copyable 
and $(N,L)$ is acyclic, a core stable feasible assignment exists and can be found in time polynomial in $p$ and~$n$.
\end{theorem}

\noindent
Now, the stalker game in Example~\ref{ex:NS:empty} does not admit a Nash stable outcome even if we make all
activities copyable. However, for copyable activities we can still find a Nash stable outcome
in polynomial time if the social network is acyclic.

\begin{theorem}\label{thm:NS:copyable}
Given an instance $(N,A,(\succeq_{i})_{i \in N},L)$ of \gGASP\ where each activity $a \in A^{*}$ is copyable 
and the graph $(N,L)$ is acyclic, one can decide whether it admits a Nash stable outcome in time polynomial in $p$ and 
$n$.
\end{theorem}
\begin{proof}
If the input graph $(N, L)$ is a forest, we can process each of its connected components separately, so we assume that 
$(N, L)$ is a tree. We choose an arbitrary node as the root and construct a rooted tree by orienting the edges in $L$ towards the leaves. Then, 
for each player $i$, each activity $a$ and each $k\in[n]$ and $t\in[k]$
we set $f_{i}((a,k),t)$ to \emph{true} if the following condition holds: 
there exists a feasible assignment $\pi$ for the subtree rooted at $i$ where $|\pi_i|=t$, $\pi(i)=a$, 
each player in $\pi_i$ likes $(a, k)$ at least as much as any alternative she can deviate to
(including the void activity), and no player who is not in $\pi_i$ has an NS-deviation.
Otherwise, we set $f_{i}((a,k),t)$ to \emph{false}.

For each player $i \in N$, each alternative $(a,k) \in X$, 
and each $t \in [k]$, we initialize $f_i((a,k),t)$ to {\em true} if $t=1$ 
and $i$ weakly prefers $(a,k)$ to any alternative of size $1$, 
and we set $f_i((a,k),t)$ to {\em false} otherwise. Then, for $i \in N$ from the bottom to the root, we iterate through all the 
children of $i$ and update $f_{i}((a,k),t)$ one by one; more precisely, for each child $j$ of $i$ and for 
$t=k,\ldots,1$, we set $f_{i}((a,k),t)$ to {\em true} if there exists an $x \in [t]$ such that both $f_{i}((a,k),x)$ and 
$f_{j}((a,k),t-x)$ are {\em true}, or $f_{i}((a,k),t)$ is {\em true} and there exists $(b,\ell) \in X$ such that 
$f_{j}((b,\ell),\ell)$ is {\em true} and neither $i$ nor $j$ want to move across the ``border'', 
i.e., $(a,k) \succeq_{i} (b,\ell+1)$ and $(b,\ell) \succeq_{j} (a,k+1)$.
It is easy to see that a Nash stable assignment exists if and only if $f_{r}((a,k),k)$ is {\em true} 
for some alternative $(a,k) \in X$, 
where $r$ is the root of the rooted tree. If this is the case, a Nash stable feasible assignment can be found using dynamic programming.
\end{proof}

\section{Hardness Results for Nash Stability}
We now move on to the case where each activity can be used at most once. We will show that computing Nash stable 
outcomes of \gGASP s is NP-complete even when the underlying network is a path, a star, or a graph with constant size 
connected components. This problem is in NP for any social network: given an assignment, 
we can easily check whether it is Nash stable. 

Our proof for paths is by reduction from a restricted version of the NP-complete problem {\sc Rainbow Matching}. Given a 
graph $G$, a {\em proper edge coloring} is a mapping $\phi:E(G) \rightarrow \calC$ where $\phi(e) \neq \phi(e^{\prime})$ 
for all edges $e, e^{\prime}$ such that $e\neq e^{\prime}$, $e\cap e^{\prime}\neq\emptyset$.
Without loss of generality, we assume that $\phi$ is surjective.
 A {\em properly edge-colored graph} $(G,\phi)$ is a graph together with a proper edge coloring. A matching $M$ 
in an edge colored graph $(G,\phi)$ is called a {\em rainbow matching} if all edges of $M$ have different colors. An 
instance of {\sc Rainbow Matching} is a graph $G$ with a proper edge coloring $\phi$ and an integer $k$. It is a 
``yes''-instance if $G$ admits a rainbow matching with at least $k$ edges and a ``no''-instance otherwise.  
\citet{Le2014} show that {\sc Rainbow Matching} remains NP-complete even for properly edge-colored paths.

\begin{theorem}\label{thm:hardness:path:NS}
Given an instance of \gGASP\ whose underlying graph is a path, it is {\rm NP}-complete to determine whether it has a Nash 
stable feasible assignment.
\end{theorem}
\begin{proof}
The hardness proof proceeds by a reduction from {\sc Path Rainbow Matching}.

Given an instance $(G, \phi, k)$ of {\sc Path Rainbow Matching} where the set of colors is given by 
$\calC$ where $|\calC|=q$, we construct an instance of \gGASP\ on a path as follows. We create a vertex 
player $v$ for each $v \in V(G)$ and an edge player $e$ for each $e \in E(G)$, and align them in the order consistent 
with $E(G)$; specifically, we let $N_{G}=V(G) \cup E(G)$ and $L_{G}=\{\, \{v,e\} \mid v \in e \in E(G) \,\}$. To the 
right of the graph $(N_G,L_G)$, we attach a path that consists of ``garbage collectors'' $\{g_{1},g_{2},\ldots,g_{q-k}\}$ 
and $q$ copies $(N_{c},L_{c})$ of the stalker game
where $N_{c}=\{c_1,c_2\}$ and $L_{c}=\{ \{c_1,c_2\}\}$ for each $c \in \calC$.
We introduce a color activity $c$ for each color $c \in \calC$.
Each vertex player $v$ approves color activities $\phi(e)$ of its adjacent edges $e$ with size $3$; each edge player $e$ 
approves the color activity $\phi(e)$ of its color with size $3$; each garbage collector $g_{i}$ approves any color 
activity $c$ with size $1$; finally, for players in $N_{c}$, $c \in \calC$,  player $c_1$ approves its color activity 
$c$ with size $1$, whereas player $c_2$ approves $c$ with size $2$.

We will now argue 
that $G$ contains a rainbow matching of size at least $k$ if and only if there exists a Nash stable 
feasible assignment.

Suppose that there exists a rainbow matching $M$ of size $k$. We construct a feasible assignment 
$\pi$ where for each $e=\{u, v\}\in M$ we set $\pi(e)=\pi(u)=\pi(v)=\phi(e)$,
each garbage collector $g_{i}$, $i\in[q-k]$, is arbitrarily assigned to one of the remaining $q-k$ color activities, 
and the remaining players are assigned to the void activity. The assignment $\pi$ is Nash stable, since every garbage 
collector as well as every edge or vertex player assigned to a color activity
are allocated their top alternative, and no remaining player has an NS feasible deviation.

Conversely, suppose that there is a Nash stable feasible assignment $\pi$. Let $M=\{\, e \in E(G) \mid \pi(e) \in \calC 
\,\}$. We will show that $M$ is a rainbow matching of size at least $k$. To see this, notice that $\pi$ cannot 
allocate a color activity to a member of $N_{c}$, since otherwise no feasible assignment would be Nash stable. Further, at 
most $q-k$ color activities are allocated to the garbage collectors, which means that at least $k$ color activities 
should be assigned to vertex and edge players. The only individually rational way to do this is to select triples of the 
form $(u,e,v)$ where $e=\{u,v\} \in E(G)$ and assign to them their color activity $\phi(e)$. Thus, $M$ is a rainbow 
matching of size at least $k$.
\end{proof}

For \gGASP s on stars we provide a reduction from the NP-complete problem {\sc Minimum Maximal 
Matching} (MMM). An instance of MMM is a graph $G$ and a positive integer $k \leq |E(G)|$. It is a ``yes''-instance if 
$G$ admits a maximal matching with at most $k$ edges, and a ``no''-instance otherwise. The problem remains NP-complete for 
bipartite graphs \cite{Demange2008}.

\begin{theorem}
	\label{thm:hardness:star:NS}
Given an instance of \gGASP\ whose underlying graph is a star, it is {\rm NP}-complete to determine whether it has a Nash 
stable feasible assignment.
\end{theorem}
\begin{proof}
To prove NP-hardness, we reduce 
from MMM on bipartite graphs. 
Given a bipartite graph $(U,V,E)$ and an integer $k$, we create a star with center $c$ and 
$|V|+1$ leaves: one leaf for each vertex player $v \in V$ plus one stalker $s$. 
We introduce an activity $u$ for each $u \in U$, and two additional activities $a$ and $b$. 
A player $v \in V$ approves $(u, 1)$ for each activity $u$ such that $\{u, v\}\in E$ as well as  
$(a, |V|-k+1)$ and prefers the former to the latter. That is, 
$(u,1) \succ_{v}(a,|V|-k+1)$ for every $u \in U$ with $\{u, v\}\in E$; $v$ is indifferent among 
the activities associated with its neighbors in the graph, that is, 
$(u,1)\sim_{v} (u^{\prime},1)$ for all $u,u^{\prime} \in U$ such that $\{u, v\}, \{u^{\prime}, v\}\in E$. 
The center player $c$ approves both $(a, |V|-k+1)$ and $(b,1)$, and prefers the former to the latter, 
i.e., $(a,|V|-k+1) \succ_{c} (b,1) \succ_{c} (a_{\emptyset},1)$. 
Finally, the stalker $s$ only approves $(b,2)$. 

We now show that $G$ admits a maximal matching $M$ with at most $k$ edges
if and only if our instance of \gGASP\ admits a Nash stable assignment.
Suppose that $G$ admits a maximal matching $M$ with at most $k$ edges. We construct a feasible assignment $\pi$ 
by setting $\pi(v)=u$ for each $\{u, v\}\in M$, assigning $|V|-k$ vertex players and the center to $a$,
and assigning the remaining players to the void activity. Clearly, the center $c$ has no 
incentive to deviate and no vertex player in a singleton coalition wants to deviate to the coalition of the center. 
Further, no vertex $v$ has an NS-deviation to an unused activity $u$, since if $\pi$ admits such a deviation, this 
would mean that $M\cup \{u,v\}$ forms a matching, a contradiction with the maximality of $M$. 
Finally, the stalker player has no incentive to deviate since 
the center player does not play $b$. Hence, $\pi$ is Nash stable.

Conversely, suppose that there exists a Nash stable feasible assignment $\pi$ and let $M=\{\, \{\pi(v),v\} \mid v \in V 
\land \pi(v) \in U \,\}$. We will show that $M$ is a maximal matching of size at most $k$. By Nash stability, the stalker 
player should not have an incentive to deviate, and hence the center player and $|V|-k$ vertex players are assigned to 
activity $a$. It follows that $k$ vertex players are not assigned to $a$, and therefore 
$|M| \leq k$. Moreover, $M$ is a matching since each vertex player is assigned to at most one activity, 
and by individual rationality 
each activity can be assigned to at most one player. Now suppose towards a contradiction that $M$ is not 
maximal, i.e., there exists an edge $\{u,v\} \in E$ such that $M\cup \{u,v\}$ is a matching. 
This would mean that under $\pi$ no player is assigned to $u$ and $v$ is assigned to the void activity; 
hence, $v$ has an NS-deviation to $u$, contradicting the Nash stability of~$\pi$.
\end{proof}

In the analysis of cooperative games on social networks one can usually assume that the social network
is connected: if this is not the case, each connected component can be processed separately.
This is also the case for \gGASP\ as long as all activities are copyable. However, if each activity
can only be used by a single group, different connected components are no longer independent,
as they have to choose from the same pool of activities. Indeed, we will now show
that the problem of finding Nash stable outcomes remains NP-hard even if the size of each connected component 
is at most four. Our hardness proof for this problem proceeds by reduction from a restricted version of {\sc 3Sat}. 
Specifically, we consider (3,B2)-{\sc Sat}: in this version of {\sc 3Sat}
each clause contains exactly $3$ literals, and each variable occurs exactly twice positively and twice 
negatively. This problem is known to be NP-complete \cite{Berman2003}.

\begin{theorem}\label{thm:hardness:smallcomponent:NS}
Given an instance of \gGASP\ where each connected component of the underlying graph has size at most $4$, it is 
{\em NP}-complete to determine whether it has a Nash stable feasible assignment.
\end{theorem}
\begin{proof}
We reduce from (3,B2)-{\sc Sat}. Consider a formula $\phi$ with variable set $X$ and clause set $C$, 
where for each variable $x\in X$ we write $x_1$ and $x_2$ for the two positive occurrences of $x$, 
and $\bar x_1$ and $\bar x_2$ for the two negative occurrences of $x$. 
For each $x\in X$, we introduce four players $x_1, x_2, {\bar x_1},{\bar x_2}$,
which correspond to the four occurrences of $x$. 
For each clause $c\in C$, we introduce one stalker $s_{c}$ 
and three other players $c_1,c_2$, and $c_3$.
The network $(N,L)$ consists of one component for each clause---a star with center $s_{c}$ and leaves  
$c_1$, $c_2$, and $c_3$---and of two components for each variable $x\in X$ consisting of a single edge each: $\{x_1, x_2\}$ and $\{\bar x_1, \bar x_2\}$.
Thus,
the size of each component of this graph is at most~$4$.

For each $x\in X$ we introduce one variable activity $x$, 
two positive literal activities $x_1$ and $x_2$, 
two negative literal activities ${\bar x_1}$ and ${\bar x_2}$, 
and two further activities $a_{x}$ and ${\bar a_{x}}$. 
Also, we introduce an activity $c$ for each clause $c \in C$. 
Thus,
\[
A^{*}=\bigcup_{x \in X}\{x,x_1,x_2,{\bar x_1},{\bar x_2},a_x,{\bar a_{x}}\} \cup C.
\]
For each $x\in X$ the preferences of the positive literal players $x_1$ and $x_2$ are given as follows:
\begin{align*}
&x_1:~(x,2) \succ (x,1) \succ (x_1,1)\succ (x_2,2) \succ (a_{x},1)\succ (a_{\emptyset},1),\\
&x_2:~(x,2)\succ (x_2,1)\succ (x_1,2) \succ (a_{x},2)\succ (a_{\emptyset},1).
\end{align*}
Similarly, for each $x\in X$ the preferences of the negative literal players ${\bar x_1}$ and ${\bar x_2}$ are given as follows:
\begin{align*}
&{\bar x_1}:~(x,2) \succ (x,1) \succ ({\bar x_1},1)\succ ({\bar x_2},2) \succ ({\bar a_{x}},1)\succ (a_{\emptyset},1),\\
&{\bar x_2}:~(x,2)\succ ({\bar x_2},1)\succ ({\bar x_1},2) \succ ({\bar a_{x}},2)\succ (a_{\emptyset},1).
\end{align*}
In a Nash stable assignment none of the activities $a_{x}$, ${\bar a_{x}}$, $a_\emptyset$ 
can be assigned to literal players. Hence, there are only two possible cases: first, both players $x_1$ and $x_2$ 
are assigned to $x$, and players ${\bar x_1}$ and ${\bar x_2}$ are assigned to ${\bar x_1}$ and ${\bar x_2}$, 
respectively; second, both players ${\bar x_1}$ and ${\bar x_2}$ are assigned to $x$, and players $x_1$ and $x_2$  
are assigned to $x_1$ and $x_2$, respectively.

For players in $N_{c}$ where $c=\ell^c_1 \lor \ell^c_2 \lor \ell^c_3$, the preferences are given by
\begin{align*}
&c_r:~(\ell^c_{r},1)\succ (c,2) \succ (a_{\emptyset},1), \qquad (r = 1,2,3)\\
&s_{c}:~(\ell^c_{1},2) \sim (\ell^c_{2},2) \sim (\ell^c_{3},2) \sim (c,2)\succ (a_{\emptyset},1).
\end{align*}
That is, players $c_1$, $c_2$, and $c_3$ prefer to engage alone in their approved literal activity, 
whereas $s_{c}$ wants to join one of the adjacent leaves whenever 
$\pi(s_c)=a_\emptyset$ and that leaf is assigned a literal activity; 
however, the leaf would then prefer to switch to the void activity. 
This means that if there exists a Nash stable outcome, 
at least one of the literal activities must be used outside of $N_{c}$, 
and some leaf and the stalker $s_{c}$ must be assigned to activity $c$.
We will show that $\phi$ is satisfied by some assignment if and only if there exists a Nash stable outcome.

Suppose that there exists a truth assignment that satisfies $\phi$. First, for each variable $x$ that is set to True, 
we assign positive literal activities $x_1$ and $x_2$ to the positive literal players $x_1$ and $x_2$, respectively, 
and assign $x$ to the negative literal players ${\bar x_1}$ and ${\bar x_2}$. For each variable $x$ that is set to False, 
we assign negative literal activities ${\bar x_1}$ and ${\bar x_2}$ to 
the  negative literal players ${\bar x_1}$ and ${\bar x_2}$, respectively, and assign $x$ to 
the positive literal players $x_1$ and $x_2$. Note that this procedure uses at least one of the literal activities $\ell^c_{1}$, 
$\ell^c_{2}$ and $\ell^c_{3}$ of each clause $c \in C$, since the given truth assignment satisfies $\phi$. Then, for each 
clause $c \in C$, we select a player $c_{j}$ whose approved activity $\ell^c_{j}$ has been assigned to some literal player, 
and assign $c_{j}$ and the stalker to $c$, and the rest of the clause players to their approved literal activity 
if it is not used yet, and to the void activity otherwise. It is easy to see that the resulting assignment $\pi$ is Nash 
stable.

Conversely, suppose that there exists a Nash stable feasible assignment $\pi$. By Nash stability, for each variable $x 
\in X$, either a pair of positive literal players $x_1$ and $x_2$ or a pair of negative literal players ${\bar 
x_1}$ and ${\bar x_2}$ should be assigned to the corresponding pair of literal activities; in addition, for each clause 
$c \in C$, the stalker $s_{c}$ and one of the players $c_1$, $c_2$, and $c_3$ should engage in the activity $c$, thereby 
implying that the approved literal activity of the respective leaf
should be assigned to some literal players. Then, take the truth assignment 
that sets the variable $x$ to True if its positive literal players $x_1$ and $x_2$ 
are assigned to positive literal activities $x_1$ and $x_2$; 
otherwise, $x$ is set to False. This assignment can be easily seen to satisfy $\phi$.
\end{proof}

\section{Fixed Parameter Tractability}
In the instances of \gGASP\ that are created in our hardness proofs, the number of activities 
is unbounded. It is thus natural to 
wonder what can be said when there are few activities to be assigned. We will show that 
for each of the restricted families of graphs considered in the previous section, 
\gGASP\ is fixed parameter tractable with respect to the number of activities. 

The basic idea behind each of the three algorithms
is that we fix a set of activities that will be assigned to the players, 
and for each possible subset $B\subseteq A^*$ of activities we 
check whether there exists a stable assignment using the activities
from that subset only. Our algorithms for paths and for small components use dynamic programming, 
allowing us to build up the set $B$ step-by-step.

We begin by giving the dynamic program that works for paths. We consider the path from left to right,
and, for each initial segment of the path, guess a set $B^{\prime} \subseteq B$ of activities that will be used in that segment of players.
For each choice of these sets, we keep track of whether it is possible to construct an assignment that does not admit an NS-deviation within the initial segment under consideration.

\begin{theorem}\label{thm:FPT:path:NS}
There exists an algorithm that, given an instance of \gGASP\ whose underlying graph is a path, 
checks whether this instance has a Nash stable feasible 
assignment and finds one if it exists, and runs in time $O(4^{p}pn(p+n))$
\end{theorem}
\begin{proof}
Suppose that $N=[n]$ and $L=\{\, \{i,i+1\} \mid i=1,2,\ldots,n-1\,\}$.
First, we guess a subset $B\subseteq A^*$ of non-void activities to be used; there are $2^p$
possibilities, so we try them all. 
For each $B$, we solve the problem by dynamic programming. 
For each $i\in[n]$, 
each $B^{\prime} \subseteq B$, 
each alternative $(a,k) \in B^{\prime} \times [n] \cup \{(a_{\emptyset},1)\}$, 
and each number $t \in [k]$,  
we let $f_{i}(B,B^{\prime},(a,k),t)$ be \emph{true} if there exists a feasible assignment 
$\pi:[i] \rightarrow B^{\prime}\cup \{a_{\emptyset}\}$
with the following properties:
\begin{itemize}
\item each activity in $B^{\prime}$ is assigned to some player in $[i]$;
\item the $t$ players in $\{i-t+1,i-t+2,\ldots,i\}$ belong to the same group as $i$ 
      and weakly prefer $(a,k)$ to $(b, 1)$ for each $b\in A\setminus B$;
\item player $i-t+1$ weakly prefers $(a,k)$ to the coalition he would end up in by joining his predecessor;
\item player $i-t$ weakly prefers his alternative at $\pi$ to $(a,k+1)$; and
\item the rest of the players in $\{1,2,\ldots,i-t\}$ weakly prefer their alternative under $\pi$ 
     to engaging alone in any of the activities in $A\setminus B$ 
     and have no NS feasible deviation to activities in $B^{\prime}$.
\end{itemize}
Otherwise, we let $f_{i}(B,B^{\prime},(a,k),t)$ be \emph{false}.

For player $i=1$, if $B^{\prime}=\{a\}$, $t=1$, and player $1$ weakly prefers $(a,k)$ to 
each alternative $(b, 1)$ such that $b\in A\setminus B$, 
we set $f_{1}(B,B^{\prime},(a,k),t)$ to {\em true} and otherwise to {\em false}.

For $i=2,\ldots,n$, $f_{i}(B,B^{\prime},(a,k),t)$ is {\em true} only if player $i$ weakly prefers $(a,k)$ 
to $(b, 1)$ for each $b\in A\setminus B$, and in addition, either
\begin{itemize}
\item $t=1$, and players $i$ and $i-1$ can be \emph{separated} from each other, 
i.e., there exists $(b,\ell) \in X$ such that $f_{i-1}(B,B^{\prime}\setminus\{a\},(b,\ell),\ell)$ is true, 
$(b,\ell) \succeq_{i-1} (a,k+1)$ and $(b,\ell+1) \succeq_{i} (a,k)$, or 
\item $t \geq 2$ and $f_{i-1}(B,B^{\prime},(a,k),t-1)$ is {\em true}.
\end{itemize}
If the condition above is not satisfied, then $f_{i}(B,B^{\prime},(a,k),t)$ is set to {\em false}. 
It is not difficult to see that a Nash stable assignment exists if and only if $f_{n}(B,B,(a,k),k)$ 
is {\em true} for some alternative 
$(a,k) \in X$ and some $B \subseteq A^{*}$. The bound on the running time is immediate.
\end{proof}

Our algorithm for networks with small connected components is similar to the dynamic program we just discussed. 
We essentially pretend that the components are arranged in a ``path'', and run the algorithm as before. Within each component, we have enough time to consider all possible assignments, allowing us to treat components as ``big vertices''. The resulting algorithm is FPT with respect to the combined parameter $p+c$, where $c$ is a bound on the size of the components of the network.

\begin{theorem}
	\label{thm:FPT:smallcomponents:NS}
	There exists an algorithm that given an instance of \gGASP\ on a graph with constant-size connected components 
	checks whether it has a Nash stable feasible assignment, finds one if it exists, 
	and runs in time $O(p^{c}8^{p}kn^2)$, where $c$ is the maximum size of a connected component 
	and $k$ is the number of connected components.
\end{theorem}
\begin{proof}
	We give a dynamic programming algorithm. Suppose our graph $(N,L)$ has $k$ connected components $(N_{1},L_{1}), 
	(N_{2},L_{2}), \ldots, (N_{k},L_{k})$. For each $i\in[k]$, each set $B \subseteq A^{*}$ of activities 
	assigned to $N$, and each set $B^{\prime} \subseteq B$ of activities assigned to $\bigcup^{i}_{j=1}N_{j}$, we let 
	$f_{i}(B,B^{\prime})$ denote whether there is such an assignment that gives rise to a Nash stable outcome. 
	Specifically, $f_{i}(B,B^{\prime})$ is \emph{true} if and only if there exists an individually rational feasible 
	assignment $\pi:\bigcup^{i}_{j=1}N_{j} \rightarrow A$ such that
	\begin{itemize}
		\item $\pi$ uses exactly the activities in $B^{\prime}$, 
		i.e., $\pi^{b}\neq \emptyset$ for all $b \in B^{\prime}$ and 
		$\pi^{b}= \emptyset$ for all $b \in A^{*}\setminus B^{\prime}$, and
		\item no player in $\bigcup^{i}_{j=1}N_{j}$ has an NS-deviation to 
		an activity in $B^{\prime}$ or to an activity in $A^*\setminus B$.
	\end{itemize}
	
	For $i=1$, each $B \subseteq A^{*}$, and each $B^{\prime}\subseteq B$, we compute the value of $f_{1}(B,B^{\prime})$ 
	by trying all possible mappings $\pi:N_1 \rightarrow B^{\prime} \cup \{a_{\emptyset}\}$, and checking whether it is an 
	individually rational feasible assignment using all activities in $B^{\prime}$ and such that no player in $N_{1}$ has 
	an NS-deviation to a used activity in $B^{\prime}$ or an unused activity in $A^*\setminus B$.
	For $i=2,3,\ldots,k$, each $B \subseteq A^{*}$, and $B^{\prime} \subseteq B$, we set 
	$f_{i}(B,B^{\prime})$ to {\em true} if there exists a bipartition of $B^{\prime}$ into $P$ and $Q$ such that $f_{i-1}(B,P)$ 
	is true and there exists a mapping $\pi:N_{i} \rightarrow Q\cup \{a_{\emptyset}\}$ such that $\pi$ is an individually 
	rational feasible assignment using all the activities in $Q$, and no player in $N_{i}$ has an NS-deviation to a used 
	activity in $Q$ or an unused activity in $A^{*}\setminus B$. It is not difficult to see that a Nash stable solution 
	exists if and only if $f_{k}(B,B)$ is {\em true} for some $B \subseteq A^{*}$. If this is the case, such a stable feasible 
	assignment can be found using standard dynamic programming techniques. The bound on the running time is immediate.
\end{proof}

For networks given by star graphs, we use a different technique to obtain an FPT result, namely (derandomized) color coding. The algorithm begins by guessing the alternative $(a,k)$ assigned to the center player. Next, we again guess the precise set $B$ of activities in use by the players not assigned to alternative $(a,k)$. We then randomly color leaf players by activities in $B$ (or by the void activity), rejecting colorings that are infeasible or must lead to NS deviations. Crucially, the latter task reduces to straightforward counting questions, which allows this method to succeed.

\begin{theorem}
	\label{thm:FPT:star:NS}
There exists an algorithm that given an instance of \gGASP\ on a star checks whether it has a Nash stable feasible 
assignment, finds one if it exists, and runs in time $O(2^pp^{p+1}n \log n)$.
\end{theorem}
\begin{proof} 
For each $(a,k) \in X$ and $B \subseteq A^* \setminus \{a\}$, we will check whether there exists a Nash stable 
assignment such that the center $c$ and $k-1$ leaves engage in $a$, exactly $|B|$ leaf players are assigned to activities in $B$, 
and the rest of the players are assigned to the void activity. We will check whether the center player $c$ weakly prefers $(a,k)$ to every alternative $(b,\ell) \in B \times \{2\}\cup (A \setminus B) \times \{1\}$. If this is the case, we will proceed; otherwise, there is no Nash stable outcome with the above properties, since the center player would have an incentive to deviate. 

Now we will check whether there is an assignment of activities in $B$ to leaf players that gives rise to a Nash stable outcome.
To this end, we use the color-coding technique to design a 
randomized algorithm. We `color' each leaf player using colors in $B$ independently and uniformly 
at random. Suppose that there exists a Nash stable assignment $\pi$ as described above. Then, the probability that the players who engage in activities from $B$
on their own (denote this set by $S$) are assigned these activities by a 
coloring $\chi$ chosen at random is $|B|^{-|B|}$: there are $|B|^{n-1}$ possible colorings, and $|B|^{n-1-|B|}$ of them coincide with $\pi$ on $S$. 
We can then derandomize our algorithm using a 
family of $k$-perfect hash functions \cite{Alon1995}.

Now, fix a coloring $\chi:N\setminus \{c\} \rightarrow B$. We seek to assign each player $i \in N\setminus\{c\}$ to one of the activities, namely either to $a$, to $\chi(i) \in B$, or to $a_{\emptyset}$, 
in such a way that exactly one agent of each color engages in the color activity 
and $k$ players including the center are assigned to $a$. We will show that there 
exists a polynomial-time algorithm that finds a Nash stable outcome compatible with $\chi$, 
or determines that no such assignment exists.
For each $b \in B$, we denote by $N_{b}=\{\, i \in N \setminus \{c\} \mid \chi(i)=b\,\}$ the set of players 
of color $b$.
 
For each color $b \in B$, each $i \in N_b$, and $\ell=0,1,\ldots,k-1$, we will first determine whether $i$ can be assigned to $b$, and exactly $\ell$ players in $N_b$ can be assigned to activity $a$; we denote this subproblem by $f_b(i,\ell)$.
We initialize $f_b(i,\ell)$ to {\em true} if $\ell=0$, player $i$ weakly prefers $(b,1)$ to every alternative $(b^{\prime},1)$ such that $b^{\prime} \in A \setminus B$, and the players $i$ and $c$ can be {\em separated} from each other, i.e., $a=a_{\emptyset}$ or $i$ weakly prefers $(b,1)$ to $(a,k+1)$.
Otherwise, we set $f_b(i,\ell)$ to {\em false}. 
We then iterate through all the players $j \in N_b \setminus \{i\}$ and update $f_b(i,\ell)$ one by one. Specifically, we set $f_b(i,\ell)$ to {\em true} if 
\begin{itemize}
\item $f_b(i,\ell-1)$ is true and player $j$ can be assigned to $a$, i.e., $j$ weakly prefers $(a,k)$ to every alternative $(b^{\prime},1)$ such that $b^{\prime} \in A \setminus B$; or
\item $f_b(i,\ell)$ is true and player $j$ can be assigned to $a_{\emptyset}$, i.e., $j$ weakly prefers $(a_{\emptyset},1)$ to every alternative $(b^{\prime},1)$ such that $b^{\prime} \in A \setminus B$ and $(a,k+1)$.
\end{itemize}
Otherwise, we set $f_b(i,\ell)$ to {\em false}. 

Now, we will determine whether there are exactly $\ell$ players in $N$ who can be engaged in $a$; we denote this subproblem by $f(\ell)$. For each $\ell \in [k]$, we initialize $f(\ell)$ to {\em true} if $\ell=1$ and $f(\ell)$ to {\em false} otherwise. Then, we iterate through all the colors $b \in B$ and update $f(\ell)$: for each $b \in B$ and each $\ell=k,k-1,\ldots,1$, we set $f(\ell)$ to {\em true} if there exists $x \in \{0,1,\ldots,k-1\}$ such that $f_b(j,x)$ is true for some $j \in N_b$ and $f(\ell-x)$ is true; otherwise, we set $f(\ell)$ to {\em false}. 

Finally, we reject the coloring if $f(k)$ is false. It is clear that the algorithm does not reject the coloring if there exists a Nash stable feasible assignment that is compatible with $\chi$. We omit the proof for the bound on the running time.
\end{proof}

\section{Core stability}
By adapting the reductions for Nash stability, we can show that checking the existence of a core stable outcome is also NP-hard. This result holds for all classes of graph families that we have considered.

\begin{theorem}\label{thm:hardness:path:core}
Given an instance of \gGASP\ whose underlying graph is a path, a star, or 
has connected components whose size is bounded by $3$, 
it is {\em NP}-complete to determine whether it has a core stable feasible assignment.
\end{theorem}
\begin{proof}
To verify that a given feasible assignment is core stable, it suffices to check that for every alternative $(a,k)$ there is no connected coalition with at least $k$ players who strictly prefer $(a,k)$ to the alternative of their current coalition. For the networks we consider this can be done in polynomial time, and hence our problem is in NP. The hardness reductions are similar to the respective reductions for Nash stability; essentially, we have to replace copies of the stalker game with copies of the game with an empty core. 

\vspace{0.3cm}
\noindent{\bf Paths}
We prove the hardness via a reduction from {\sc Path Rainbow Matching}. Given such an instance $(G, \phi, k)$ where $|\calC|=q$, we first construct the graph $(N_{G},L_{G})$ as defined in the proof of Theorem \ref{thm:hardness:path:NS}. To the right of the graph $(N_{G},L_{G})$, we attach a path that consists of garbage collectors $\{g_{1},g_{2},\ldots,g_{q-k}\}$ and $q$ copies $(N_{c},L_{c})$ of empty-core instances of Example \ref{ex:core:empty} where $N_{c}=\{c_1,c_2,c_3\}$ and
$L_{c}=\{ \{c_1,c_2\},\{c_2,c_3\}\}$ for each $c \in \calC$. 
For each color $c \in \calC$, we introduce a color activity $c$, and introduce additional activity $a_{c}$, which will be used only among players in $N_{c}$; specifically, we set $A=\bigcup_{c \in \calC}\{c, a_{c}\} \cup \{a_{\emptyset}\}$.

Each vertex player $v \in V(G)$, edge player $e \in E(G)$ and garbage collector $g_{i}$ has the same approval preference as in the proof of Theorem \ref{thm:hardness:path:NS}. The preference for players in $N_{c}$ $(c \in \calC)$ is cyclic and given by 
\begin{align*}
&c_1:~(c,2)\succ (a_{c},3) \succ  (a_{\emptyset},1),\\
&c_2:~(a_{c},2) \succ (c,2) \succ (a_{c},3) \succ (a_{\emptyset},1),~\mbox{and}\\
&c_3:~(a_{c},3)\succ (c,1) \succ (a_{c},2) \succ (a_{\emptyset},1).
\end{align*}

We will now argue that $G$ contains a rainbow matching of size at least $k$ if and only if there exists a core stable feasible assignment. 

Suppose that there exists a rainbow matching $M$ of size $k$. We construct a feasible assignment 
$\pi$ where for each $e=\{u, v\}\in M$ we set $\pi(e)=\pi(u)=\pi(v)=\phi(e)$,
each garbage collector $g_{i}$, $i\in[q-k]$, is arbitrarily assigned to one of the remaining $q-k$ color activities, each triple of $c_1$, $c_2$, and $c_3$ $(c \in \calC)$ is assigned to $a_{c}$, 
and the remaining players are assigned to the void activity. 
The assignment $\pi$ is core stable, since every garbage 
collector as well as every edge or vertex player assigned to a color activity
are allocated their top alternative, and no connected subsets of the remaining players together with color activity can strongly block $\pi$.

Conversely, suppose that there is a core stable feasible assignment $\pi:N \rightarrow A$. Let $M=\{\, e \in E(G) \mid \pi(e) \in \calC \,\}$. We will show that $M$ is a rainbow matching of size at least $k$. To see this, notice that at $\pi$, all the color activities should be played outside $N_{c}$'s, since otherwise no core stable assignment would exist as we have seen in Example \ref{ex:core:empty}. Further, at most $q-k$ colour activities are played among the garbage collectors, which means that at least $k$ colour activities should be assigned to vertex and edge players. Again, the only individual rational way to do this is to select triples of the form $(u,e,v)$ where $e=\{u,v\} \in E(G)$ and assign to them their colour activity $\phi(e)$; thus, $M$ is a rainbow matching of size at least $k$.

\vspace{0.3cm}
\noindent{\bf Stars}
We reduce from a restricted variant of MMM where the graph is a bipartite graph. 

Given a bipartite graph $(U,V,E)$ and an integer $k$, we create a star with center $c$ and the $|V|+2$ leaves: one leaf for each vertex $v \in V$ plus two other players $s_{1}$ and $s_{2}$. We then introduce an activity $u$ for each $u \in U$, and three other simple activities $a$, $x$, and $y$. 

A player $v \in V$ approves $(u,1)$ for each $u \in U$ such that $\{u,v\}\in E$ as well as $(a,|V|-k+1)$ and prefers the former to the latter.  That is, $(u,1) \succ_{v}(a,|V|-k+1)$ for any $u \in U$ with $\{u,v\}\in E$; $v$ is indifferent among activities associated with its neighbors in the graph, that is, $(u,1)\sim_{v} (u^{\prime},1)$ for all $u,u^{\prime} \in U$ such that $\{u,v\},\{u^{\prime},v\}\in E$. The center player $c$ strictly prefers $(a,|V|-k+1)$ to any other alternative, and has the same cyclic preferences over the alternatives of $x$ and $y$ as in Example \ref{ex:core:empty} together with players $s_{1}$ and $s_{2}$, given by
\begin{align*}
&s_{1}:~(y,2) \succ (x,3) \succ (a_{\emptyset},1)\\
&c: (a,|V|-k+1) \succ (x,2) \succ (y,2) \succ (x,3) \succ (a_{\emptyset},1)\\
&s_{2}:~(x,3) \succ (y,1) \succ (x,2) \succ (a_{\emptyset},1).
\end{align*}
Here, $s_{1}$'s (respectively, the center $c$ and the player $s_{2}$) preference corresponds to the one for player $1$ (respectively, player $2$ and player $3$) in Example \ref{ex:core:empty}. We will show that $G$ contains a maximal matching of size at most $k$ if and only if there exists a core stable feasible assignment in a similar manner to the previous proof.

Suppose that $G$ admits a maximal matching $M$ with at most $k$ edges. We construct a feasible assignment $\pi$ by setting $\pi(v)=u$ for each $\{u,v\} \in M$, and assigning $|V|-k$ non-matched vertex players and the center to $a$, assigning $s_{2}$ to $y$, and assigning the remaining players to the void activity. The center $c$ is allocated to her top alternative, and hence no connected subset of the three players $s_{1},c$ and $s_{2}$ together with activity $x,y$, and $z$ strongly blocks $\pi$. Further, no vertex player $v$ together with an unused vertex activity $u$ strongly blocks $\pi$, since if such a pair $\{u,v\}$ existed, this would mean that $\{u,v\}$ is not included in $M$, and hence $M\cup \{u,v\}$ forms a matching, which contradicts the maximality of $M$. Hence, $\pi$ is core stable.

Conversely, suppose that there exists a core stable feasible assignment $\pi$ and let $M=\{\, \{\pi(v),v\} \mid v \in V \land \pi(v) \in U \,\}$. We will show that $M$ is a maximal matching of size at most $k$. By core stability, the center player and $|V|-k$ vertex players are assigned to the activity $a$; otherwise, no core stable outcome would exist as we have seen in Example $2$; thus, $|M| \leq k$. Notice further that $M$ is a matching since each vertex player plays at most one activity, and by individual rationality each vertex activity should be assigned to at most one player. Now suppose towards a contradiction that $M$ is not maximal, i.e., there exists an edge $\{u,v\} \in E$ such that $u \in U$, $v \in V$, and $M\cup \{u,v\}$ is a matching. This would mean that $\pi$ assigns no player to $u$ and no vertex activity to $v$, and hence the coalition $\{v\}$ together with the vertex activity $u$ strongly blocks $\pi$, contradicting the stability of $\pi$.

\vspace{0.3cm}
\noindent{\bf Small Components}
We reduce from (3,B2)-{\sc Sat}. Consider a formula $\phi$ with variable set $X$ and clause set $C$, where for each variable $x \in X$ we write $x_1$ and $x_2$ for the two positive occurrences of $x$, and ${\bar x_1}$ and ${\bar x_2}$ for the two negative occurrences of $x$. Corresponding to the variable occurrences, we introduce four players $x_1$, $x_2$, ${\bar x_1}$, and ${\bar x_2}$ for each variable $x \in X$. We also introduce two other players $x$ and ${\bar x}$ for each variable $x \in X$. For each clause $c \in C$, we introduce three players $c_1,c_2$, and $c_3$. 
The network consists of one component for each clause $c \in C$: a star with center $c_2$ and leaves $c_1$ and $c_3$, and of two components for each variable $x \in X$: a star with center $x$ and leaves $x_1$ and $x_2$, and a star with center ${\bar x}$ and leaves ${\bar x_1}$ and ${\bar x_2}$. Hence, the size of each connected component of this graph is at most $3$. 
We then construct the set of activities given by
\[
A^*=\bigcup_{x \in X} \{x,x_1,x_2,{\bar x_1},{\bar x_2},a_{x},b_{x},{\bar a_{x}},{\bar b_{x}}\}.
\]
For each $x \in X$, the preferences of the positive literal players $x_1$ and $x_2$ and the positive variable player $x$ are given as follows: 
\begin{align*}
&x_1:~(x,3) \sim (x_{1},1)\succ (b_x,2) \succ (a_{x},3)\succ (a_{\emptyset},1)\\
&x_2:~(x,3) \sim (x_{2},1) 
\succ (a_{x},2)\succ (b_{x},2) \succ (a_{x},3) \succ (a_{\emptyset},1)
\\
&x:~(x,3) \succ (a_{x},3)\succ (b_{x},1) \succ (a_{x},2) \succ (a_{\emptyset},1)
\end{align*} 
Similarly, for each $x \in X$, the preferences of the negative literal players $\bar x_1$ and $\bar x_2$, and the negative variable player $\bar x$ are given as follows:
\begin{align*}
&{\bar x_1}:~(x,3) \sim ({\bar x_1},1)\succ ({\bar b_{x}},2) \succ ({\bar a_{x}},2) \succ (a_{\emptyset},1)\\
&{\bar x_2}:~(x,3) \sim ({\bar x_2},1) 
\succ ({\bar a_{x}},2)\succ ({\bar b_{x}},2) \succ ({\bar a_{x}},3) \succ (a_{\emptyset},1)
\\
&{\bar x}:~(x,3) \succ ({\bar a_{x}},3)\succ ({\bar b_{x}},1) \succ ({\bar a_{x}},2) \succ (a_{\emptyset},1)
\end{align*}
Notice that the preferences of each triple contains a cyclic relation, and hence in a core stable assignment, there are only two possible case: first, all the three players $x_1$, $x_2$, and $x$ are assigned to $x$, and players ${\bar x_1}$, ${\bar x_2}$, and ${\bar x}$ are assigned to activities ${\bar x_1}$ , ${\bar x_2}$, and ${\bar b_{x}}$, respectively; second, all the players ${\bar x_1}$, ${\bar x_2}$, and ${\bar x}$ are assigned to $x$, and players $x_1$, $x_2$, and $x$ are assigned to activities $x_1$, $x_2$, and $b_{x}$, respectively.

For each clause $c \in C$ where $c=\ell^c_1 \lor \ell^c_2 \lor \ell^c_3$, the preferences for clause players $c_1$, $c_2$, and $c_3$ are again cyclic and given as follows: 
\begin{align*}
&c_1:~(\ell^c_1,2)\succ (a_{\emptyset},1),\\
&c_2:~(\ell^c_2,2)\succ (\ell^c_1,2) \succ (\ell^c_3,2) \succ (a_{\emptyset},1),\\
&c_3:~(\ell^c_3,2)\succ (\ell^c_1,1) \succ (\ell^c_2,2) \succ (a_{\emptyset},1).
\end{align*}
If there exists a core stable outcome, it must be the case that at least one of the literal activities $\ell^c_1$, $\ell^c_2$, and $\ell^c_3$ must be used outside of the three players $c_1$, $c_2$, and $c_3$; otherwise, no feasible assignment would be core stable.

Now we will show that $\phi$ is satisfied by some truth assignment if and only if there is a core stable feasible assignment. 

Suppose that there exists a truth assignment that satisfies $\phi$. We construct a core stable feasible assignment $\pi$ as follows. First, for each variable $x$ that is set to True, we assign positive literal activities $x_1$, $x_2$, and $b_{x}$ to positive literal players $x_1$, $x_2$, and a variable player $x$, respectively, and assign a variable activity $x$ to players ${\bar x_1}$, ${\bar x_2}$, and ${\bar x}$. For each variable $x$ that is set to False, we assign negative literal activities ${\bar x_1}$, ${\bar x_2}$, and ${\bar b_{x}}$ to negative literal players ${\bar x_1}$, ${\bar x_2}$, and a variable player ${\bar x}$, respectively, and assign a variable activity $x$ to players $x_1$, $x_2$, and $x$. Note that this procedure uses at least one of the literal activities $\ell^c_1$, $\ell^c_2$ and $\ell^c_3$ of each clause $c$, since the given truth assignment satisfies $\phi$. Then, for each clause $c \in C$, we assign activities by constructing a digraph with vertices being potential assignments and identifying an ``undominated'' activity. Let 
\[
V_{c}=\{(\ell^c_1,\{c_1,c_2\}),(\ell^c_2,\{c_2,c_3\}),(\ell^c_1,\{c_3\}),(\ell^c_3,\{c_2,c_3\})\},
\]
and delete from $V_{c}$ vertices whose activities $\ell^{c}_{j}$ are already assigned to some players. Then, orient from $(a,S)$ to $(b,T)$ if there exists a common player $c_j \in S\cap T$ who strictly prefers $(a,|S|)$ to $(b,|T|)$. This digraph is acyclic since at least one of the literal activities $\ell^c_1$, $\ell^c_2$, and $\ell^c_3$ has been already assigned to some literal players. If the digraph is empty, i.e., all the activities are already assigned in a previous step, we assign the void activity to all the clause players $c_1$, $c_2$, and $c_3$. Otherwise, we assign the activity of a source vertex of the digraph to its coalition and the void activity to the rest. The resulting assignment $\pi$ of players to activities is core stable, because no variable and literal player wishes to change their alternative and no connected subset of each $N_{c}$ cannot strongly block $\pi$.

Conversely, suppose that there exists a core stable feasible assignment $\pi$. By core stability, for each variable $x \in X$, either a pair of positive literal activities $x_1$ and $x_2$ or a pair of negative literal activities ${\bar x_1}$ and ${\bar x_2}$ should be assigned to the corresponding pair of literal players; further, for each clause $c$, at least one of the literal activities $\ell^c_1$, $\ell^c_2$, and $\ell^c_3$ should be played outside of the clause players $c_1$, $c_2$, and $c_3$. Then, take the truth assignment that sets the variables $x$ to True if their positive literal players $x_1$ and $x_2$ are assigned to positive literal activities $x_1$ and $x_2$; otherwise, $x$ is set to False; this can be easily seen to satisfy $\phi$.
\end{proof}

\noindent
Our FPT result for graphs with small connected components
can also be adapted to the core. 
In contrast, our approach for Nash stability for paths and stars does not seem to generalize to core stability, and we leave these cases for future work.
                                                     
\begin{theorem}\label{thm:FPT:smallcomponents:core}
There exists an algorithm that given an instance of \gGASP\ 
checks whether it has a core stable feasible assignment, finds one if it exists, and runs in time 
$O(p^{c+1}8^{p}kn^2)$, where $c$ is the maximum size of the connected components and 
$k$ is the number of connected components.
\end{theorem}
\begin{proof}
We give a dynamic programming. Suppose our graph $(N,L)$ has $k$ connected components $(N_{1},L_{1}), (N_{2},L_{2}), \ldots, (N_{k},L_{k})$. For each $i=1,2,\ldots,k$, each set $B \subseteq A^{*}$ of activities assigned to $N$, and each set $B^{\prime} \subseteq B$ of activities assigned to $\bigcup^{i}_{j=1}N_{j}$, we denote by $f_i(B,B^{\prime})$ whether there is an assignment of $\bigcup^{i}_{j=1}N_{j}$ that gives rise to a core stable outcome. 
Specifically, $f_{i}(B,B^{\prime})$ is {\em true} if and only if there exists an individually rational feasible assignment $\pi:\bigcup^{i}_{j=1}N_{j} \rightarrow A$ such that 
\begin{itemize}
\item $\pi$ only uses the activities in $B^{\prime}$, i.e., $\pi^{b}\neq \emptyset$ for all $b \in B^{\prime}$ and $\pi^{b}=\emptyset$ for all $b \in A^*\setminus B^{\prime}$, and 
\item no connected subset $S \subseteq \bigcup^{i}_{j=1}N_{j}$ together with activity in $B^{\prime}\cup(A^{*} \setminus B)$ strongly blocks $\pi$. 
\end{itemize}

For $i=1$, each $B \subseteq A^{*}$, and $B^{\prime}\subseteq B$, we compute the value of $f_1(B,B^{\prime})$ by trying all possible mappings $\pi:N \rightarrow B^{\prime} \cup \{a_{\emptyset}\}$, and check whether it is an individually rational feasible assignment using all activities in $B^{\prime}$ and such that no connected subset $S \subseteq N_{1}$ together with activity in $B^{\prime}\cup(A^{*} \setminus B)$ strongly blocks $\pi$. For $i\geq 2$ from $i=2,3,\ldots,k$, each $B \subseteq A^{*}$, and $B^{\prime} \subseteq B$, we set $f_{i}(B,B^{\prime})$ to true if there exists a bipartition of $B^{\prime}$ into $P$ and $Q$ such that $f_{i-1}(B,P)$ is true and there exists an individually rational feasible assignment $\pi:N_{i} \rightarrow Q\cup \{a_{\emptyset}\}$ such that each activity in $Q$ is assigned to some player in $N_i$, and no connected subset $S \subseteq N_{i}$ together with activity in $Q\cup(A^{*} \setminus B)$ strongly blocks $\pi$. It is not difficult to see that a core stable solution exists if and only if $f_{k}(B,B)$ for some $B \subseteq A^{*}$. If this is the case, such a stable feasible assignment can be found using standard dynamic programming techniques. The bound on the running time is immediate.
\end{proof}


\section{Conclusion}
In this paper, we have initiated the study of group activity selection problems with network structure, and found that 
even for very simple families of graphs computing stable outcomes is NP-hard. We identified several ways 
to circumvent this computational intractability. For \gGASP s with copyable activities, we showed that there 
exists a polynomial time algorithm to compute stable outcomes, and for \gGASP s with few activities, we 
provided fixed parameter algorithms for restricted classes of networks.

We leave several interesting questions for future work. Our fixed-parameter tractability results can be extended to 
more general graph families, such as graphs with bounded pathwidth and graphs with a bounded number of internal nodes. 
However, for general graphs, the exact parameterized complexity of determining the existence of stable outcomes is 
unknown. When the underlying graph is complete, one can adapt techniques of \citet{Darmann2012} to show that the 
problem of computing Nash stable outcomes is in XP with respect to $p$; for other networks, including trees, it is not even clear 
whether our problem is in XP with respect to $p$. It would be also interesting to investigate the parameterized 
complexity of \gGASP s using other parameters.

Another promising research direction is to study analogues of other solution concepts from the hedonic games literature for 
\gGASP s; in particular, it would be interesting to understand the complexity 
of computing individually stable outcomes in \gGASP s.

\bibliographystyle{aaai}

\begin{thebibliography}{}

	\bibitem[\protect\citeauthoryear{Alon, Yuster, and Zwick}{1995}]{Alon1995}
	Alon, N.; Yuster, R.; and Zwick, U.
	\newblock 1995.
	\newblock Color-coding.
	\newblock {\em J. ACM} 42(4):844--856.
	
	\bibitem[\protect\citeauthoryear{Aziz and Savani}{2016}]{Aziz2016}
	Aziz, H., and Savani, R.
	\newblock 2016.
	\newblock Hedonic games.
	\newblock In Brandt, F.; Conitzer, V.; Endriss, U.; Lang, J.; and Procaccia,
	A.~D., eds., {\em Handbook of Computational Social Choice}. Cambridge
	University Press.
	\newblock chapter~15.
	
	\bibitem[\protect\citeauthoryear{Berman, Karpinski, and
		Scott}{2003}]{Berman2003}
	Berman, P.; Karpinski, M.; and Scott, A.~D.
	\newblock 2003.
	\newblock Approximation hardness of short symmetric instances of {MAX}-3{SAT}.
	\newblock Technical Report 049.
	\newblock http://eccc.hpi-web.de/report/2003/049/.
	
	\bibitem[\protect\citeauthoryear{Chalkiadakis, Greco, and
		Markakis}{2016}]{Chalkiadakis2016}
	Chalkiadakis, G.; Greco, G.; and Markakis, E.
	\newblock 2016.
	\newblock Characteristic function games with restricted agent interactions:
	Core-stability and coalition structures.
	\newblock {\em Artificial Intelligence} 232:76--113.
	
	\bibitem[\protect\citeauthoryear{Darmann \bgroup et al\mbox.\egroup
	}{2012}]{Darmann2012}
	Darmann, A.; Elkind, E.; Kurz, S.; Lang, J.; Schauer, J.; and Woeginger, G.
	\newblock 2012.
	\newblock Group activity selection problem.
	\newblock In {\em Proceedings of the 8th International Conference on Internet
		and Network Economics}, WINE 2012,  156--169.
	
	\bibitem[\protect\citeauthoryear{Darmann}{2015}]{Darmann2015}
	Darmann, A.
	\newblock 2015.
	\newblock Group activity selection from ordinal preferences.
	\newblock In {\em Proceedings of the 4th International Conference on
		Algorithmic Decision Theory}, ADT 2015,  35--51.
	
	\bibitem[\protect\citeauthoryear{Demange and Ekim}{2008}]{Demange2008}
	Demange, M., and Ekim, T.
	\newblock 2008.
	\newblock Minimum maximal matching is {NP}-hard in regular bipartite graphs.
	\newblock In {\em Proceedings of the 5th International Conference on Theory and
		Applications of Models of Computation}, TAMC 2008,  364--374.
	
	\bibitem[\protect\citeauthoryear{Demange}{2004}]{Demange2004}
	Demange, G.
	\newblock 2004.
	\newblock On group stability in hierarchies and networks.
	\newblock {\em Journal of Political Economy} 112(4):754--778.
	
	\bibitem[\protect\citeauthoryear{Elkind}{2014}]{Elkind2014}
	Elkind, E.
	\newblock 2014.
	\newblock Coalitional games on sparse social networks.
	\newblock In {\em Proceedings of the 10th International Conference on Internet
		and Network Economics}, WINE 2014,  308--321.
	
	\bibitem[\protect\citeauthoryear{Igarashi and Elkind}{2016}]{Igarashi2016}
	Igarashi, A., and Elkind, E.
	\newblock 2016.
	\newblock Hedonic games with graph-restricted communication.
	\newblock In {\em Proceedings of the 15th International Conference on
		Autonomous Agents and Multiagent Systems}, AAMAS 2016,  242--250.
	
	\bibitem[\protect\citeauthoryear{Le and Pfender}{2014}]{Le2014}
	Le, V.~B., and Pfender, F.
	\newblock 2014.
	\newblock Complexity results for rainbow matchings.
	\newblock {\em Theoretical Computer Science} 524(C):27--33.
	
	\bibitem[\protect\citeauthoryear{Myerson}{1977}]{Myerson1977}
	Myerson, R.~B.
	\newblock 1977.
	\newblock Graphs and cooperation in games.
	\newblock {\em Mathematics of Operations Research} 2(3):225--229.
	
\end{thebibliography}

\end{document}